\documentclass[a4paper, 12pt]{article}
\usepackage{indentfirst}
\usepackage{overpic}
\usepackage{multirow}
\usepackage[hang]{subfigure}
\usepackage[mathscr]{euscript}

\usepackage{srcltx, amsmath, amssymb, amsthm}
\usepackage{amssymb}
\usepackage{graphicx}
\usepackage{color}
\usepackage{bm}
\usepackage{dsfont}
\usepackage{epstopdf}

\theoremstyle{plain}
\newtheorem{thm}{Theorem}
\newtheorem{coro}{Corollary}
\newtheorem{lem}{Lemma}

\newtheorem{pf}{Proof of Theorem}

\newcommand\e{\epsilon}

\newcommand\mR{\mathds {R}}

\newcommand\lam{\lambda}

\title{Adaptive elastic net and Separate Selection from Least Squares for ultra-high dimensional regression models}

\author{Yuehan Yang\thanks{School of Statistics and Mathematics, Central University of Finance and Economics, Beijing 100081, email: {\tt yyh@cufe.edu.cn}.}, 
	~~ Hu Yang\thanks{College of Mathematics and Statistics, Chongqing University, Chongqing 401331, PR China, email: {\tt yh@cqu.edu.cn}.}
}

\begin{document}
\maketitle







\begin{abstract}

This paper studies the asymptotic properties of the adaptive elastic net in ultra-high dimensional sparse linear regression models and proposes a new method called SSLS (Separate Selection from Least Squares) to improve prediction accuracy. Besides, we prove that SSLS has the superior performance both in the theoretical part and empirical part.

In this paper, we prove that the probability of adaptive elastic net selecting wrong variables can decays at an exponential rate with very few conditions. Irrepresentable Condition or similar constraint isn't necessary in our proof. We derive accurate bounds of bias and mean squared error (MSE) which both depend on the choice of parameters, and also show that there exists a bias of asymptotic normality of the adaptive elastic net. Furthermore, simulations and empirical part both show that the prediction accuracy of the penalized least squares requires more improvement.

Therefore, we propose SSLS to improve the prediction. It selects variable first, reducing high dimension to low dimension by using the adaptive elastic net in this paper. In the second step, the coefficients are constructed based on the OLS estimation. We show that the bias of SSLS can decays at an exponential rate. Also, MSE decays to zero. Finally, we prove that the variable selection consistency of SSLS implies the asymptotic normality of SSLS. Simulations given in this paper illustrate the performance of the SSLS, adaptive elastic net and other penalized least squares. The index tracking problem in stock market is studied in the empirical part with other methods.
\end{abstract}

\vspace*{4mm} 
\noindent {\bf Keywords:} Adaptive Elastic Net; SSLS; Variable selection; Oracle property.



\section{Introduction}
In recent years, modern technology makes massive, large-scale data sets appear frequently. That is, the number of parameters ($p$) is much larger than the sample size ($n$). Financial problems for instance, investment portfolio involves hundreds of stocks but valid sample sizes are often only one hundred or less since the samples obtained before 6 months ago often loses their effectiveness. Moreover, computational field, biological field, etc, data sets like this ($n \ll p$) is becoming more and more important in diverse fields, and poses great challenges and opportunities for statistical analysis.

Consider the regression model
\begin{equation}
\label{eqn:regressionmodel}
y_n = X_n \beta_n + \epsilon_n,
\end{equation}
where $X_n $ is the $n\times p$ design matrix of predictor variables. $\beta_n \in \mathds {R}^p$ is the true regression coefficients and $\epsilon_n = (\epsilon_{1,n},\epsilon_{2,n},...,\epsilon_{n,n})'$ is a vector of i.i.d. random variables with mean $0$ and variance $\sigma^2$.

Increasing statistic tools are developed to solve the high-dimensional data analysis, \cite{Efron07,fan2001variable,bickel2009simultaneous,lee2010adaptive,Tibshirani1996,zou2009adaptive}. Penalized least squares like lasso, \cite{YuBin06(lasso)} established the Irrepresentable Conditions for the variable selection \cite{fuchs2004recovery, meinshausen2006high, tropp2004greed, wainwright2009sharp}; elastic net \cite{Zou2005(elastic)}, adaptive lasso \cite{HuangJian08(adaptive),Zou06(adaptive)}, etc have been widely used.

SCAD \cite{fan2001variable} is also a very popular method due to its good computational and statistical properties. It enjoys the oracle property\footnote{ Oracle property can correctly identify the set of nonzero components of $\beta_n$ with probability tending to $1$, and at the same time, estimate the nonzero components accurately \cite{chatterjee2013rates,fan2001variable}. } 
which means it can perform as well as the oracle. \cite{fan2011nonconcave} studied the penalized likelihood with the $l_1$-penalty. \cite{bradic2011penalized} proposed the penalized composite likelihood method in ultra-high dimensions. \cite{belloni20111$l_1$} discussed the $l_1$-penalized quantile regression in high-dimensional sparse models. \cite{fan2014adaptive} proposed weighted robust lasso in the ultra-high dimensional setting that the number of parameters grow exponentially with the sample size.

The adaptive elastic net estimator is defined as the minimizer of the weighted $l_1$-penalized and $l_2$-penalized least squares criterion function.
\begin{equation}
\label{eqn:adaptiveelastic}
  \hat \beta_n(\lam_{1,n},\lam_{2,n}) =  \mathop {argmin }\limits_{\beta \in \mR^p}\{\frac{1}{2}||y_n-X_n\beta||^2_2+\frac{1}{2}\lam_{2,n} ||\beta||^2_2 + \lam_{1,n} \sum\limits_{j=1}^{p} \hat w_{j,n} |\beta_j|\}.
\end{equation}
The $l_1$ part performs automatic variable selection, while the $l_2$ part stabilizes the solution paths and improves the prediction. $\{\hat w_{j,n}\}^p_{j=1}$ are the adaptive data-driven weights, which used to reduce the bias problem induced by the $l_1$-penalty. Hence the adaptive elastic net is an improved version of the lasso, elastic net and adaptive lasso. Adaptive weights can be computed by different values: $\hat w_j = (|\hat \beta_j (ols)|)^{- \gamma}$, $j =1,...,p$ where $\gamma$ is a positive constant \cite{Zou06(adaptive)}, $\hat w_{j,n} =|\tilde{\beta}_{j,n}|^{-1}$ and $\tilde{\beta}_{n}= X'_n y_n / n$ \cite{HuangJian08(adaptive)}, $\hat w_{j,n} = (|\hat \beta_{j,n} (elastic net)|)^{- \gamma}$ \cite{zou2009adaptive}. The adaptive elastic net method is shown that which enjoys the oracle property \cite{fan2001variable} with a diverging number of predictors \cite{zou2009adaptive}.

Although the oracle property of the adaptive elastic net estimators with a diverging number of predictors was already studied before, the asymptotic properties of the adaptive elastic net with the ultra-high dimensional setting remains unknown. Furthermore, penalized least squares always need the particular condition to get variable selection consistency and few literatures discussed about the accuracy of this statistical inference on the nonzero regression parameters before.

In this paper, we first study the asymptotic properties of adaptive elastic net for the growing number of parameters where the dimensionality can grow exponentially with the sample size. We find a simple set
\begin{equation}
A_{n}  \equiv \left\{ ||W_n||_{\infty}  \leqslant K n^{\eta}\right\},
\end{equation}
where $W_n = X'_n \e_n/ \sqrt{n}$ . $\eta$ is a positive constant. We compute adaptive parameter $\widetilde{\beta}_n$ by the lasso estimator and the adaptive weighter $\hat w$ is computed as $\hat w_{j,n} = |\widetilde{\beta}_{j,n}|^{\gamma}$, $\gamma \geqslant 0$.  According to the estimation consistency of adaptive parameter, the choice of $\gamma$ and conditional on $\{A_n\}$, we lead to variable selection consistency of the adaptive elastic net when the noise vector $\e_n$ has i.i.d. entries in the ultra-high dimensional setting. In our proof, the probability for adaptive elastic net to select true model is covered by the probability of $A_n$. The first half part of the proof of Theorem \ref{thm:1} states the probability of $A_n^c$ decays at an exponential rate under ultra-high dimensional setting. The latter part states the relationship between $P(\hat S_n \neq S_n)$ and $P(A_n)$ without any other constrains.

Then, we introduce the MSE and bias of adaptive elastic net and indicate that their decay rate depends on the probability of selecting wrong variables $P(\hat S_n \neq S_n)$. Consequently, the MSE and bias can both decay to zero with suitable choice of tuning parameters $\lam_{1,n}$ and $\lam_{2,n}$. However, one weakness of these rates is that they may lead to an inferior rate depending on the choice of the tuning parameters and the initial parameter $\tilde{\beta}_n$.

We also find that the traditional penalized least squares cannot have an ideal prediction accuracy both in simulations and financial fields. For instance, we apply penalized least square method to track SP500. It has $2\%$ to $4\%$ predicted (annual) tracking errors when select 50 constituent stocks. If we reduce the number of selected stocks like 20, the tracking errors increase significantly. We want to improve the mentioned theoretical defect and prediction accuracy, oscillation simultaneous by applying other method.

Therefore, we propose a valid technique, called SSLS, for Separate Selection from Least Squares. It selects variable first and sets others to $0$, reducing high dimensional setting to low dimension setting by adaptive elastic net in the paper, then uses Ordinary Least Squares (OLS) to estimate coefficients. That is
\begin{equation}
\hat \beta_{j,n}(ssls) = \left\{ {\begin{array}{cc}
\hat \beta_{j,n}(ols), & j \in \hat S_n \\
0, & j \notin \hat S_n,
\end{array}} \right.
\end{equation}
where $\hat \beta_{j,n}(ols)$ obtained in the low dimensional linear regression models: $(y_n, X_{\hat S_n})$. There are two reasons why the ordinary least squares (OLS) estimates is unsuitable in high dimensional setting: prediction accuracy and interpretation \cite{Tibshirani1996}. But if we don't need shrink any coefficient to $0$ and consider the regression model in low dimension setting. OLS estimates lead to a satisfying prediction accuracy. This method is similar as OLS post-Lasso estimator \cite{belloni2013least} which is shown at least as well as Lasso.

We use adaptive elastic net to select the variable and study the properties of SSLS, hence SSLS has variable selection consistency. We show that the bias of SSLS decays at an exponential rate and the decay rate of MSE achieves the oracle convergence rate. Also, the asymptotic normality of SSLS is proved.

Finally, simulations and empirical part show that SSLS produces large improvement compared with other methods. In the simulation part, we implement five methods under different settings and use $l_1$, $l_2$ loss to be measures. SSLS has the best performance among others in all the settings based on $100$ replications. Similarly in empirical part, SSLS also outperforms lasso in the most months and significantly reduces the tracking error when select very few consistent stocks to track the index.

The rest of the paper is organized as follows. In section 2, we state the regularity conditions and introduce the theoretical framework, then derive the accurate convergence rate of the adaptive elastic net's probability of variable selection, the bounds of bias, MSE and the rate of convergence to the oracle distribution. Section 3 proposes a new method called SSLS and study the properties. Computations are given in Section 4. Section 5 and Section 6 show simulation examples and applications, index tracking in financial field.

\section{Model Selection Oracle Property}
We are interested in the sparse modeling problem where the true model has a sparse representation. That is, let $S_n\equiv\{j \in \{1,2,...,p\}: \beta_{j,n} \neq 0\}$ with assumption of cardinality $|S_n|=q$ ($q \ll p$). The adaptive elastic net yields an estimator $\hat S_n \equiv \{j \in \{1,2,...,p\}: \hat \beta_{j,n} \neq 0\}$. Without loss of generality, assume $\beta_n =(\beta_{1,n},...,\beta_{q,n},\beta_{q+1,n},...,\beta_{p,n})'$ where $\beta_{j,n} \neq 0$ for $j=1,...,q$ and $\beta_{j,n}=0$ for $j =q+1,...,p$. Then write $\beta^{(1)}_{n}=(\beta_{1,n},...,\beta_{q,n})'$ and $\beta^{(2)}_n=(\beta_{q+1,n},...,\beta_{p,n})$, $X_n(1)$ and $X_n(2)$ are the first $q$ and last $p-q$ columns of $X_n$ respectively. $C_n=\frac{1}{n}X_n'X_n$ can be expressed in a block-wise:
\[
C_n = \left( {\begin{array}{*{20}c}
   {C_{11,n}} \quad {C_{12,n}}\\
   {C_{21,n}} \quad {C_{22,n}}  \\
\end{array}} \right),
\]
and $W_n=X'_n \e_n/ \sqrt{n}$. Similarly, $W^{(1)}_n$ and $W^{(2)}_n$ indicate the first $q$ and last $p-q$ elements of $W_n$.

We want to use the OLS estimator to be the initial estimator $\widetilde{\beta}_n$. However $X_n'X_n$ is always singular and the OLS estimator of $\beta_n$ is no longer uniquely defined. In this case, we apply the lasso estimator $\hat \beta_{lasso}$\footnote{
	The lasso estimator is defined as
	\begin{equation}
	\label{eq:lasso}
	\hat \beta_n(\lambda_n)\in  \mathop {argmin }\limits_{\beta\in \mR^p }\{\frac{1}{2n}||Y_n-X_n\beta||^2_2+\lambda_n ||\beta||\},
	\end{equation}
	where the lasso estimator is written as $\hat \beta_{lasso}$ in this paper.
}  to be the initial estimator. 

According to the estimation consistency of $\hat \beta_{lasso}$ (related result is offered in the Lemma 2 of Appendix), we lead to variable selection consistency of the adaptive elastic net under follow constrains.

Let $\Lambda_{min}(C_{11,n})$ denotes the smallest eigenvalues of $C_{11,n}$, we define the following regular conditions

\begin{enumerate}
\item[(C.1)]  Suppose $\Lambda_{min}(C_{11,n}) > K n ^{-a}$ for some $K \in (0, \infty)$ and $a \in[0,1]$. Furthermore, $n^{-1} \sum_{i=1}^{n}x_{i,j}^2 \leqslant 1/\sigma^2$ for $j=1,...,p$.
\end{enumerate}
\begin{enumerate}
	\item[(C.2)] Restricted Eigenvalue (RE) condition, i.e. there exists constant $\kappa_\iota$, such that
	\begin{equation}
		\label{eqn:RE}
		\frac{||X_n\beta_n||^2_2}{n}\geqslant \kappa_\iota ||\beta_n||^2_2  \ \  \forall\beta_n\in R^p, \ \sum_{j \notin S_n} |\beta_{j,n}| \leqslant 3 \sum_{j\in S_n} |\beta_{j,n}|.
	\end{equation}
\end{enumerate}

(C.1) gives the regularity conditions on the design matrix, which are typical assumptions in sparse linear regression literature, see for example \cite{huang2008asymptotic,fan2014adaptive,YuBin06(lasso)}. The first part of condition (C.1) ensures a lower bound on the smallest eigenvalue of $C_{11,n}$. The second part is needed for Bernstein's inequality in Theorem \ref{thm:1}.

RE condition \eqref{eqn:RE}, developed by \cite{bickel2009simultaneous}, is a mild condition and has been studied in past work on Lasso \cite{lv2009unified}. We use $\hat \beta_{lasso}$ to be the adaptive estimator of adaptive elastic net. This condition is applied to make sure the estimation consistency of lasso estimator. 

As mentioned in the Introduction part, the choice of adaptive estimator $\widetilde{\beta}_n$ is not unique. We know there must be other more optimal estimator than the lasso estimator. For instance, if $p < n$, $\hat \beta_{OLS}$ is a more appropriate choice. However, considering about the ultra-high dimensional setting and the existing choice in literature. We prefer the lasso estimator since the related results (like the estimation consistency) of lasso is mature enough. 

\subsection{Oracle Regularized Estimator}
In this section, we study the variable selection property of adaptive elastic net when the dimensionality can grow exponentially with the sample size. That is, $P(\hat S_n =S_n) \rightarrow 1$ as $n \rightarrow \infty$ when $p=O(e^{n^c})$.

One defined sign consistency which stronger than the usual selection consistency, i.e. $P(sign(\hat \beta_n)= sign( \beta_n)) \rightarrow 1$\cite{YuBin06(lasso)}. It can be satisfied if follow inequality holds.
\begin{equation}
sign( \beta^{(1)}_n) (\hat \beta^{(1)}_n -\beta^{(1)}_n) > - |\beta^{(1)}_n|.
\end{equation}

That is, by adding a simple restraint, $|\hat \beta^{(1)}_n - \beta^{(1)}_n| < |\beta^{(1)}_n|$, we can obtain the sign consistency when the adaptive elastic net achieves the variable selection consistency. We proof the probability of selecting wrong variables here mostly for simplicity of presentation.

\begin{thm}
\label{thm:1}
Assume $\e_{i,n}$ are i.i.d. random variables with mean $0$, and variance $\sigma^2$, let $A_{n}  \equiv \{ ||W_n||_{\infty}  \leqslant K n^{\eta}\}$, where $\eta$ is a positive constant. If $\lam_{1,n} < K(\delta,\lam_{2,n}) \cdot n^{\eta+a+b}$, where $0 <\delta <1$. Then let $\eta$ bounded by
\begin{equation}
0< \eta < \dfrac{\gamma-1}{2(\gamma+1)} - \dfrac{3b+2a}{2(\gamma+1)},
\end{equation}
where $b$ is a positive constant by setting in $q=O(n^b)$ and $p = O(e^{n^c})$, $c \leqslant \frac{2}{3}\eta$. Under condition (C.1) and (C.2), we have 
\begin{equation}
P( \hat S_n= S_n ) \leqslant 1- P(A_n^c) \leqslant 1- o(e^{-n^{c}}) \rightarrow 1  \ as \ n \rightarrow \infty.
\end{equation}
\end{thm}

If $\lam_{1,n} \geqslant K(\delta,\lam_{2,n}) \cdot n^{\eta+a+b}$, we have the follow corollary
\begin{coro}
Follow the same setting in the Theorem 1 and consider the rest of $\lam_{1,n}$ that $\lam_{1,n} \geqslant K(\delta,\lam_{2,n}) \cdot n^{\eta+a+b}$, then let $\eta$ bounded by
\begin{equation}
 n^{\eta \gamma} < n^{\frac{\gamma-1-b}{2}} ||\beta^{(1)}_n||.
\end{equation}
we have $P( \hat S_n= S_n ) \leqslant 1- P(A_n^c) \leqslant 1- o(e^{-n^{c}}) \rightarrow 1  \ as \ n \rightarrow \infty$.
\end{coro}

Mention that $\eta$ and $\delta$ both are instruments help our proof but not a restraint for adaptive elastic net to select the true variables. For the choice of $\lam_{1,n}$, we should mention that under the setting of the Theorem \ref{thm:1}, $\lam_{1,n}$ is not decay to zero when $n$ tends to infinity. Beyond that, there's no other special constraint on the parameters $\lam_{1,n}$, $\lam_{2,n}$, $q$ and $p$. Therefore, Theorem~\ref{thm:1} shows that adaptive elastic net can select the true variables for most ultra-high dimensional data.

Compared with other penalized least squares, \cite{YuBin06(lasso)} proved that Irrepresentable Condition is almost necessary and sufficient for Lasso to select the true variable both in the classical setting and high-dimensional setting. In this paper, we don't need similar conditions. One of the other improvement of our technical is that, we don't need control the size of $\lam_{1,n}$ and $\lam_{2,n}$ to obtain this property.

Similar, we also can obtain the variable selection consistency for adaptive lasso by using the similar technique in the proof for proving Theorem \ref{thm:1}, which is also an improvement over literatures, e.g. \cite{HuangJian08(adaptive)} proved the variable selection consistency with so many constrains like adaptive Irrepresentable Condition. We prefer adaptive elastic net to adaptive lasso since only $l_1$ penalization method may have poor performance where there are highly correlated variables in the predictor set.

Now we introduce the bounds of bias and MSE of the adaptive elastic net:
\begin{thm}
\label{thm:2}
Assume $\e_{i,n}$ are i.i.d. random variables with mean $0$ and variance $\sigma^2$, under condition (C.1), the following bounds hold,
\begin{align}
||E\hat \beta_{n} - \beta_n||_2^2 \leqslant & 2[1+ 3 P(\hat S_n \neq S_n)] \cdot (Kn^{1-a}+\lam_{2,n})^{-2} \cdot \notag\\
&(\lam_{2,n}^2 ||\beta_n||^2_2 + \lam_{1,n}^2 E||\hat w_{n}||^2_2) \cdot \notag\\
&+ 6 P(\hat S_n \neq S_n) \cdot (||\beta_n||^2_2 + n( n \vee q) ), 
\end{align}
and
\begin{align}
E||\hat \beta_{n}- \beta_n||^2_2  \leqslant & 3[1+ 2 \sqrt{P(\hat S_n \neq S_n)}] \cdot (Kn^{1-a}+\lam_{2,n})^{-2} \cdot \notag\\
&(\lam_{2,n}^2 ||\beta_n||^2_2 + \lam_{1,n}^2 E||\hat w_{n}||^4_2+q\cdot n) \cdot \notag\\
&+ 8 \sqrt{P(\hat S_n \neq S_n)} \cdot (||\beta_n||^2_2+n^2),
\end{align}

For simplicity of presentation, let $\Lambda_{min}(\hat S_n)$ denotes the smallest eigenvalues of $\frac{1}{n} X'_{\hat S_n}X_{\hat S_n}$ and suppose $\Lambda_{min}(\hat S_n) > Kn^{-a}$. Then by choosing suitable parameter we have
\begin{equation}
||E\hat \beta_{n} - \beta_n||_2^2 \rightarrow 0 \ as \ n \rightarrow \infty,
\end{equation}
\begin{equation}
E||\hat \beta_{n}- \beta_n||^2_2  \rightarrow 0 \ as \ n \rightarrow \infty.
\end{equation}
\end{thm}

In the ultra-high dimensional setting, bias is not the only consideration of estimates. Regularization has been a popular technique which results in a reduced MSE. However, if two estimators have the same MSE, we prefer the unbiased one. To the best of our knowledge, above bounds are the smallest one among literatures about penalized least squares. Similar results can hardly obtain in other penalized least squares without the adaptive weights $\hat w$. Hence Theorem \ref{thm:2} makes adaptive elastic net very applicable.

\subsection{Rate of Convergence to the Oracle Distribution}

In this part, we investigate the rate of convergence of adaptive elastic net estimator to the oracle distribution. Let $T_n= \sqrt{n}D_n(\hat \beta_n-\beta_n)$ where $D_n$ is a $p_0 \times p$ matrix with $tr(D_n D'_n)=O(1)$. $p_0$ is an integer which can bigger than $q$ but not depending on $n$. The main result of this part gives upper and lower bound on the accuracy of approximation by the limiting oracle distribution for the adaptive elastic net. To show this property of adaptive elastic net, we need more conditions:

\begin{enumerate}
\item[(C.3)] $\max \{|\beta_{j,n}|: j \in S_n\} = O(1) $ and $\min\{|\beta_{j,n}|: j \in S_n \} \geqslant K n^{-e}$, for some $K \in (0, \infty)$ and $e \in [0, 1/2)$, such that $a+2e \leqslant 1$, where $a$ is set in (C.1)(i).
\item[(C.4)] There exists $m \in (0,1)$ and $n > m^{-1}$.

(i) $\sup \{ \alpha' D^{(1)}_n (C^{-1}_{11,n}(\lam_{2,n})C_{11,n}C^{-1}_{11,n}(\lam_{2,n})) (D_n^{(1)})' \alpha,  \ \ \forall ||\alpha||^2_2 =1    \} < m^{-1}$.,

(ii) $\inf \{ \alpha' D^{(1)}_n (C^{-1}_{11,n}(\lam_{2,n})C_{11,n}C^{-1}_{11,n}(\lam_{2,n})) (D_n^{(1)})' \alpha, \  \ \forall ||\alpha||^2_2 =1    \} > m$,
\end{enumerate}
then
$\frac{\lam_{1,n}}{\sqrt{n}} \leqslant m^{-1} n^{-m} \min \left\{  \frac{n^{-e\gamma}}{q}, \frac{n^{-e\gamma - \frac{a}{2}}}{\sqrt{q}}, n^{-a} \right\}$ and
 $$\frac{\lam_{1,n}}{\sqrt{n}} \cdot n^{\gamma/2} \geqslant m n^m \max \left\{ n^a q, q^{3/2} n^{e(1-\gamma)^{+}} \right\}.$$

(C.3) assumes that the nonzero coefficients are not masked by the estimation error, which makes it possible to separate out the signal from the noise by the adaptive elastic net. The first two bounds of (C.4) require the maximum and the minimum eigenvalues of the $p_0 \times p_0$ matrix are bounded away from zero and infinity. Other two inequalities are applied for the Edgeworth expansion results for the adaptive elastic net estimator.

Then we have the following result:

\begin{thm}
\label{thm:3}
Under conditions (C.1), (C.3) and (C.4), choose suitable $\lam_{2,n}$ to make the smallest eign-values of $C_{11,n}(\lam_2)$ greater than $K n^{-a}$ and assume that $\max\limits_{1 \leqslant j \leqslant q} c^{j,j}_{11,n}=O(1)$ where $c^{j,j}_{11,n}$ is the $(j,j)$th element of $C^{-1}_{11,n}$. Then the rate of convergence to the oracle distribution can be given as follow
\begin{equation}
\sup\limits_{B \in C_{p_0}} | P(T_n \in B)- \Phi (B, \sigma^2 Q_n )| = O(n^{-1/2} +||b_n|| + \lam_n n^{a+e(\gamma+1)/ n}),
\end{equation}
where $b_n = D^{(1)}_n C^{-1}_{11,n}(\lam_{2,n})s^{(1)}_n$,  $s^{(1)}_n$ is a $q \times 1$ vector with $j$th component $s_{j,n} = sign( \beta_{j,n})|\beta_{j,n}|^{- \gamma}$, $1 \leqslant j \leqslant q$. 
\end{thm}

Theorem \ref{thm:3} indicates that the adaptive elastic net has a bias may lead to an inferior rate converging to the limiting normal distribution. The rate critically depends on the choices of the parameters.

\section{SSLS}

Compared with the adaptive elastic net, this section proposes a valid inference procedure for both selection and estimation. We propose SSLS (Separate Selection from Least Squares) to improve the accuracy of prediction and fitting result and  show that: (i) SSLS's biases decays at an exponential rate, which much faster than original penalized least squares. Also, the MSE of SSLS can achieve at the oracle rate. (ii) We already know that adaptive elastic net has a bias of rate of convergence to the oracle distribution. In this part, SSLS estimator is proved have asymptotic normality. (iii) Furthermore, simulation and empirical part show that SSLS have much smaller fitted and predicted error compared with other methods.

Similar setting as above, let $\hat S_n= \{j \in \{ 1,2,...,p \} : \hat \beta_{j,n} \neq 0  \}$ where $\hat \beta_n$ is the adaptive elastic net estimator. Then we use OLS to estimate the surplus low-dimension set as
\begin{equation}
\hat \beta_n(ssls)= \mathop{argmin}\limits_{\beta_{\hat S_n^c}=0} ||y_n -X_n \beta||^2_2.
\end{equation}

$\hat S_n$ is obtained by the variable selection method (adaptive elastic net in this paper). When the first part of SSLS get variable selection consistency under conditions, SSLS clearly achieve the variable selection consistency. We show the follow result for SSLS using the adaptive elastic net as the first step. 
\begin{coro}
Assume $\e_{i,n}$ are i.i.d. random variables with mean $0$ and variance $\sigma^2$, under condition (C.1)...., the adaptive elastic net has variable selection consistency. That is
\begin{equation}
P( \hat S_n =S_n) \leqslant 1 - o(e^{-n^c}) \rightarrow 1 \ as \ n \rightarrow \infty.
\end{equation}

Follow the definition of SSLS estimator $\hat \beta_n(ssls) $, $\hat \beta_n(ssls) $ has the variable selection consistency too.
\end{coro}

Using the same notations as above, we show asymptotic normality of SSLS as follow.
\begin{thm}
\label{thm:4}
Assume $\e_{i,n}$ are i.i.d. random variables with mean $0$, variance $\sigma^2$ and $E[|\e|^{2+\delta}] < \infty $ for some $\delta >0$. Under condition (C.1), the variable selection of adaptive elastic net holds. Let $\Sigma_{S_n}=X'_{S_n}X_{S_n}$ and $\lim\limits_{n \rightarrow \infty} \Sigma_{S_n}^{-1} \cdot \max\limits_{i=1,...,n} \sum\limits_{j=1}^{q}x_{ij}^2=0$. Then SSLS are asymptotically normal, that is,
\begin{equation}
\alpha' \Sigma^{1/2}_{S_n}(\hat \beta_n(ssls)- \beta_n) \rightarrow_d N(0,\sigma^2),
\end{equation}
where $ \alpha$ is a vector of norm $1$.
\end{thm}

Theorem \ref{thm:4} states that the variable selection consistency of adaptive elastic net implies the asymptotic normality of SSLS estimator. Finally, we provide the general bounds for bias and MSE:
\begin{thm}
\label{thm:final}
Assume $\e_{i,n}$ are i.i.d. random variables with mean $0$ and variance $\sigma^2$, under condition (C.1), then the bias and MSE of SSLS estimator satisfy
\begin{equation}
||E\hat \beta_n (ssls)- \beta_n||^2_2 \leqslant 2 P ( \hat S_n \neq S_n  ) ( 2||\beta_n||^2_2 + \sigma^2K^{-1} n^{a-1} + \sigma^2 K^{-1} n^{a-1} n \vee q ) ,
\end{equation}
\begin{align}
E||\hat \beta_n (ssls)- \beta_n||^2_2
& \leqslant \sigma^2K^{-1} n^{a+b-1} + 8 \sqrt{P(\hat S_n \neq S_n)} (||\beta_n||^2_2 + \sigma^2  \cdot K^{-1}n^a ).
\end{align}
\end{thm}
Theorem \ref{thm:final} states that the bias of SSLS estimator decays at an exponential rate. Considering the MSE of SSLS estimator, $P(\hat S_n \neq S_n)$ decays at an exponential rate, hence it is completely determined by $\sigma^2 K^{-1} n^{a+b-1}$ which corresponds to the oracle convergence rate and cannot be improved any more.

\section{Computations}

In this section we discuss the computational issues about SSLS. We use adaptive elastic net to select the variables, hence the first half computation of SSLS is solve the adaptive elastic net estimator by LARS algorithm \cite{Efron}. The computation details are given as follow which omit the proof.
\\

\textbf{Algorithm 1} (The algorithm for the SSLS)

\begin{enumerate}
\item Given $y_n$, $X_n$ and $\lam_{2,n}$, define the predictor matrix

$$\widetilde{X}_n= \left[ {\begin{array}{*{20}c}
   X_n\\
   \lam_2 I  \\
\end{array}} \right] \in \mR^{(n+p) \times p},$$
and
$$\widetilde{y}_n = (y_n,0) \in \mR^{n+p}.$$

\item Let
$$\widetilde{X}_{j,n}(ada) = \widetilde{X}_j \times |\widetilde{\beta}_{j,n}|^{\gamma}, \ \  where \  \ \widetilde{\beta}_{j,n}\ is \ the \ initial \ estimator$$

\item Apply LARS algorithm to choose the nonzero coefficient set $\hat S_n$ by data $\widetilde{X}_n(ada)$ and $\widetilde{y}_n$.

\item Assume the linear regression model
$$y_n=X_n \beta_{\hat S_n}+ \epsilon,$$
where $\beta_{\hat S_n}=\left\{  \beta_{j,n}, j \in \hat S_n \right\}$, and solve the OLS estimator $\hat \beta_{\hat S_n} (ols)$.
\end{enumerate}

After transform $X_n$ and $y_n$ into $\widetilde{X}_{j,n}(ada)$ and $\widetilde{y}_n$, the LARS algorithm is used to compute the solution path in step 3. It is a popular and efficient algorithm hence we used in this paper.

The final step is easy but important. The estimator $\hat \beta_{\hat S_n}(ols)$ obtained by OLS estimation can get small error as much as possible, and the solution is also sparse since we get the sparse active set $\hat S_n$ in the previous step.

\section{Simulation}

Through simulations we investigate the performance of adaptive elastic net and SSLS, starting with the comparison between the adaptive lasso and lasso with Irrepresentable Condition holds or not, and then considering the performance of SSLS compared with others.

We only give a simple high-dimensional setting example in the simulation part since after this part we also investigate the performance of the SSLS which applying into the financial field compared with the traditional penalized least square method. The empirical analysis part can be seen as a more challenging scenarios.

\subsection{Adaptive Elastic Net}

To assess the performance of the adaptive elastic net estimator, we simulate data from the linear regression model
\begin{equation}
y= X' \beta + \e,
\end{equation}
where $p=200$, $n=100$ and the true regression coefficients are set as follow
\begin{equation}
\beta=\{9,6,0,....,0\},
\end{equation}
where only the first two items are nonzero. We generate i.i.d. random variables $X_{i,1},....,X_{i,199}$, $e_i$ and $\e_i$ from Gaussian distribution with mean $0$ and variance $\sigma^2$ for simplicity of presentation. $X_{i,200}$ is generated as
\begin{equation}
\label{eq:x2}
X_{i,p}= 1/6X_{i,1}+ 5/6X_{i,2}+1/2 X_{i,3}+1/6 e_{i}.
\end{equation}

According to the notations above, setting $\lam_2=1000$ and $\widetilde{\beta} =X'y/n$. We get different solution paths from the lasso and adaptive elastic net (as illustrated by Figure \ref{fig:12}). One can easily obtain that this setting doesn't satisfy the Irrepresentable Condition and hence lasso cannot select variables correctly in Figure \ref{fig:2}. As a contrast, Figure \ref{fig:1} shows the adaptive elastic net path correctly selects the true variables.

\begin{figure}[!ht]
\centering
\subfigure[Adaptive Elastic Net]{\label{fig:1}
\includegraphics[width=0.3\textwidth]{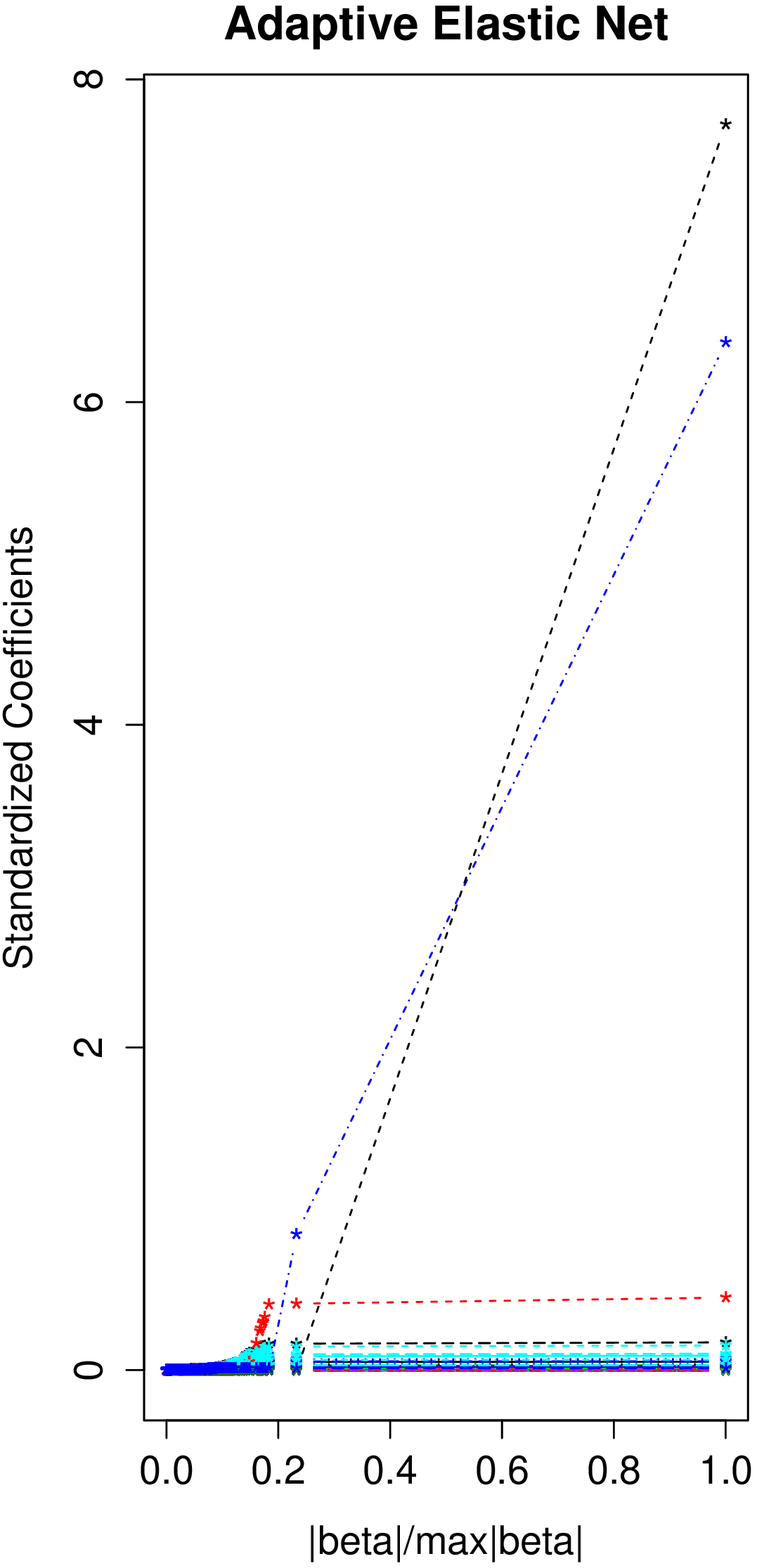}}
\subfigure[Lasso]{\label{fig:2}
\includegraphics[width=0.3\textwidth]{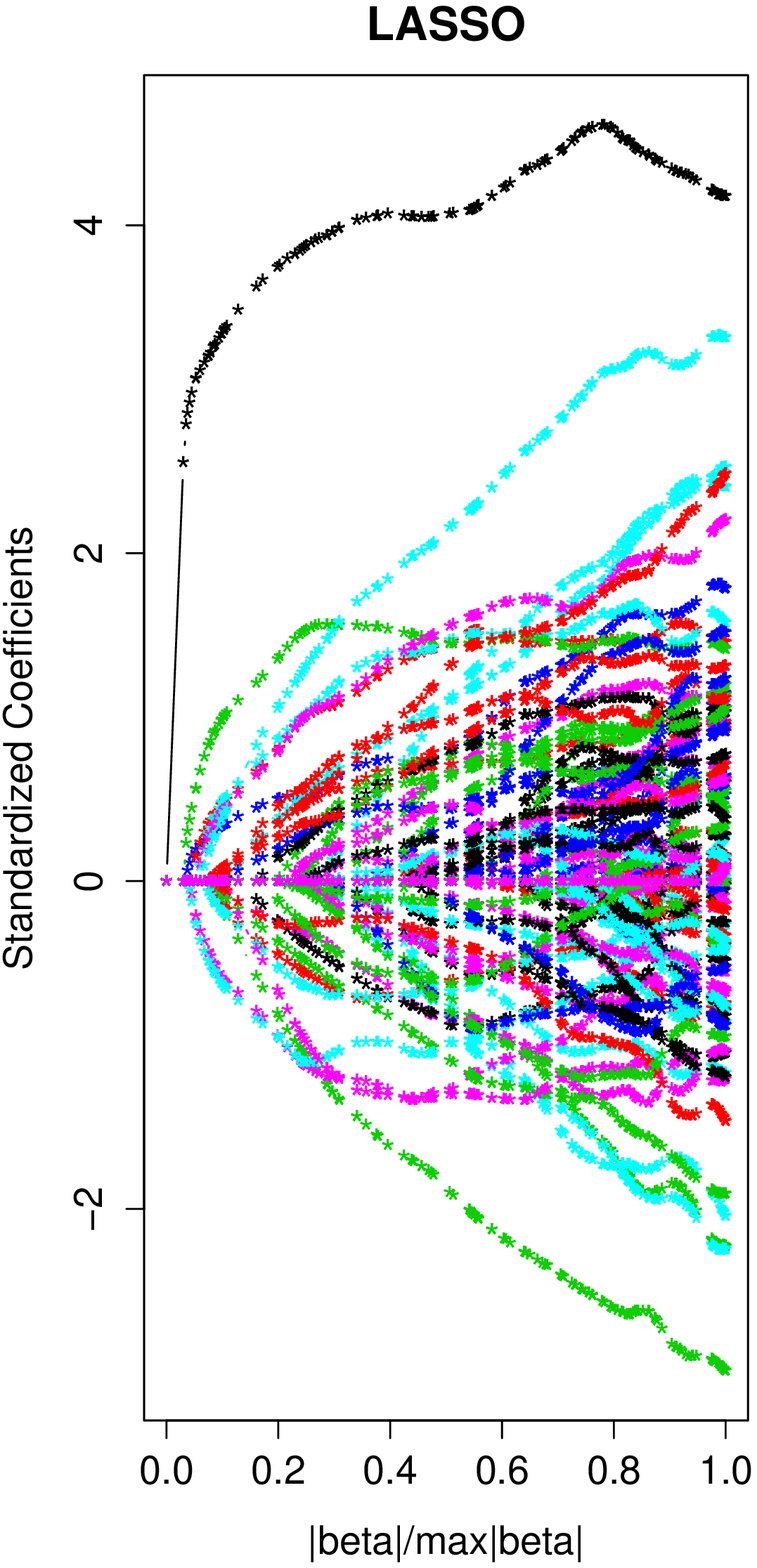}}
\caption{The adaptive elastic net solution path and the lasso solution path when Irrepresentable Condition fails}
\label{fig:12}
\end{figure}

\subsection{SSLS}

To assess the performance of the SSLS estimator and compare it with other methods, we implement five methods under two different dimensional settings (low dimensional setting vs high dimensional setting):
\begin{enumerate}
\item Lasso, the penalized least squares estimator with $l_1$ penalty proposed by \cite{Tibshirani1996}.
\item Elastic net,  the least squares estimator with both $l_1$ and $l_2$ penalty \cite{Zou2005(elastic)}.
\item Adaptive lasso, penalized least squares method with an adaptive data-driven weights \cite{Zou06(adaptive)}.
\item Adaptive elastic net, a combination of elastic net and the adaptive lasso \cite{zou2009adaptive}.
\item SSLS, separate selection from least squares which defined in Section 3.
\end{enumerate}

We simulate data from linear regression model with fixed true regressions as $\beta=\{9,6,0,....,0\}$ no matter low dimension ($p=400$, $n=100$) or high dimension ($p=10$, $n=100$). $X$ is generated from $N(0,\Sigma)$. Correlation of the covariates matrices $\Sigma$ are chosen to be (1) identity ($\Sigma=I$) and (2) generated with correlation $\rho=0.5$, $\Sigma_{i,j}=0.5^{|i-j|}$. We choose suitable tuning parameter for elastic net and adaptive elastic net to select variables for SSLS. $\lam_{2,n}$ is selected in 20 different values and we find that relatively small values (like 0.01, 0.0001 and so on) for $\lam_{2,n}$ lead to a better prediction result than larger one (like 10, 100 and so on).

Two measures are calculated: (1) $l_1$ loss: $||\hat \beta - \beta||_1$ and (2) $l_2$ loss: $||\hat \beta-\beta||_2^2$. For each design, we run the simulation 100 times and present the average of the performance measure. For simplicity of presentation, we write AEN for adaptive elastic net in the table. As depicted in Table \ref{table:1}, one should compare the performance between each method. This comparison reflects the effectiveness of SSLS deals with whether low dimensional or high dimensional setting. Furthermore, comparing the behavior of each method in each design.

It is seen that SSLS has the best performance among others in all of four settings. Beyond that, adaptive elastic net and adaptive lasso outperform lasso and elastic net in almost settings. Furthermore, SSLS has significantly lower $l_1$ and $l_2$ loss no matter smaller model size or bigger one. Adaptive elastic net has good performance in the ideal setting like $\Sigma=I$ or low dimensional setting. But in the last model, both $l_1$ and $l_2$ loss have significantly increase. Adaptive lasso has the similar behavior.
\begin{table}[!ht]
  \centering
  \caption{The $l_1$ loss and $l_2$ loss results based on 100 replications.}
  \label{table:1}
\begin{tabular}{|c|c|ccc|}
\hline
&  Model    & $l_1$ norm & $l_2$ norm & \\
\hline
  Low dimension & $p=10$  ,    $n=100$ , $\Sigma=I$ &&&    \\
    \hline
& SSLS   &  0.6253& 0.3123&\\
& AEN   &  0.7387 &  0.4181&\\
& Lasso       &   1.5547& 1.6724 &\\
& Adaptive lasso   &  0.7249 & 0.4048&\\
& Elastic Net&   1.5660& 1.6412&\\

\hline
 &  $p=10$    ,  $n=100$ ,$\rho=0.5$ && &   \\
    \hline
& SSLS   & 0.8103& 0.5106&\\
 & AEN &   1.1308 & 1.0130  &\\
& Lasso    &      1.7977& 1.6172 &\\
 & Adaptive lasso   & 1.1226 & 0.9996&\\
 & Elastic Net&   1.9065& 2.6151  &\\

\hline
High dimension &  $p=400$    ,  $n=100$ , $\Sigma=I$ && &   \\
    \hline
& SSLS     &  0.7210& 0.3955&\\
& AEN &  1.1948 &  1.0508 &\\
 & Lasso    &   2.8008& 4.4897 &\\
& Adaptive lasso  & 1.1706 & 1.0113 &\\
& Elastic Net&  2.8388&  4.6071&\\

\hline
  & $p=400$     , $n=100$ , $\rho=0.5$ && &   \\
    \hline
& SSLS    &  0.7377& 0.4391&\\
& AEN &   1.9635 &  2.6211 &\\
& Lasso      &   3.6076& 7.4180 &\\
& Adaptive lasso   &  1.9480& 2.5853&\\
& Elastic Net&   3.7928&  8.1253&\\
\hline
\end{tabular}
\end{table}

\section{Empirical Analysis: Index Tracking}

We now focus on the application of penalized least squares and SSLS in financial modeling. The performances of the fitted and predicted results are tested when they are applied to track index. In this part, we first give a brief introduction of index tracking and conduct a linear regression model for the data from stock market.

Index tracking is one of the most popular topic in the financial field. It aims to replicate the movements of an index and is the core of the index fund. Furthermore, index tracking attempts to match the performance of index as closely as possible with as small subset as possible. Thus the statistical modeling built for index tracking is a typical high dimensional model. One suitable and successful approach who can leads to sparse solutions is necessary for index tracking.

The measure for index tracking, called (annual) tracking error, is a measure of the deviation of the return of replication from target index:
\begin{equation}
TrackingError_{year}= \sqrt{252} \times \sqrt{\frac{\sum(err_t -mean(err))^2}{T-1}},
\end{equation}
where $mean(err_t)$ is the mean of $err_t$, $t=1,...,T$ and $err_t=r_t-\hat r_t$. $r_t$ is the daily return.

Our data set consists of the prices of stocks in SP500. The data come from Wind Information Co., Ltd (Wind Info), from Jan. 2012 to Dec. 2013. We divide the data set by time window: five month's data used for modeling (train set = $98$) and one month's data used for forecasting (test set = $20$). $X_{i,t}$ and $y_t$ represent the prices of the $i$th constituent stock and the index, respectively. The relationship between $X_{i,t}$ and $y_t$ can be described by a linear regression model:
\begin{equation}
y_t=\sum\limits_{i}^{500} X_{i,t}\beta_{i} + \e_t,
\end{equation}
where $\beta_i$ is the weight of the $i$th chosen stock which sparse and unknown. $\e_t$ is the error term. We need to get the estimation of $\beta_i$ by applying statistical technical.
According to the notation, we can find that tracking index topic can be seen as a high-dimensional problem which $n=98$ or $n=20$, $p=500$. We don't use cross validation to obtain the suitable number of nonzero coefficients cause there always exists practical demand about the number of selected stocks in stock market.

We use SSLS to track the Index in the next part and use tracking error to be the performance measure to show the superiority of SSLS.

\subsection{Empirical Result}

We first show the fitted and predicted results under different number of selected stocks (50 VS 20) by using SSLS. In the Figure \ref{fig:34}, Nov. 2013 is chosen to be the prediction month and the previous five months are chosen to modeling. It is seen that Figure \ref{fig:4} get better performance than Figure \ref{fig:3}. That is, reducing the number of selected stocks should slightly increase the errors. Similar, varying the length of time segments should change the tracking results too.

\begin{figure}[!ht]
\centering
\subfigure[Select 20 stocks]{\label{fig:3}
\includegraphics[width=0.48\textwidth]{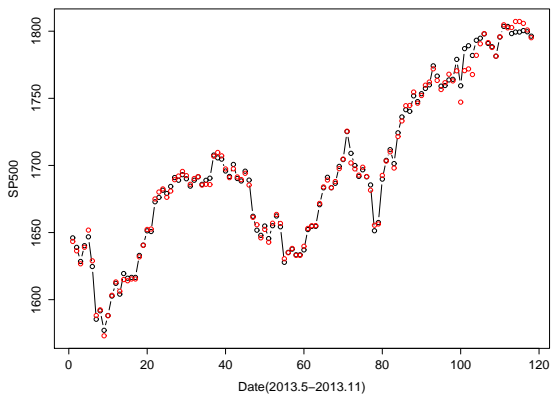}}
\subfigure[Select 50 stocks]{\label{fig:4}
\includegraphics[width=0.48\textwidth]{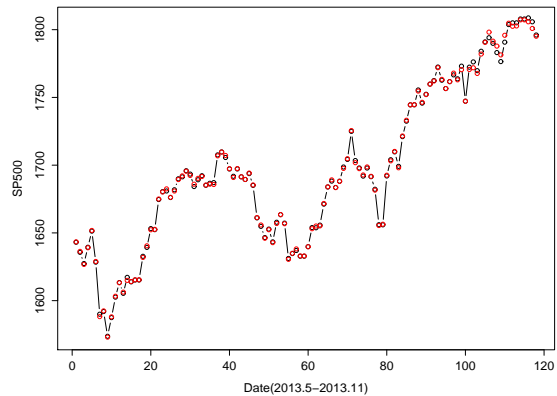}}
\caption{Select 20 stocks compared with select 50 stocks.}
\label{fig:34}
\end{figure}

Next, we select 50 constituent stocks and get the estimation of their weights in both modeling part and forecasting part by using SSLS in Figure \ref{fig:56}. As it is observed in the Figure \ref{fig:56}, fitted results are better than predicted results.

\begin{figure}[!ht]
\centering
\subfigure[Fitted results when select 50 stocks]{\label{fig:5}
\includegraphics[width=0.48\textwidth]{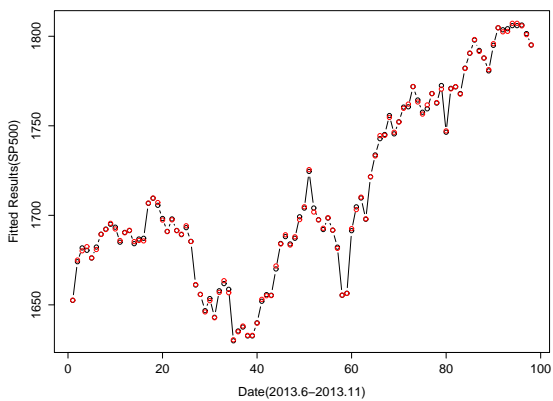}}
\subfigure[Predicted results when select 50 stocks]{\label{fig:6}
\includegraphics[width=0.48\textwidth]{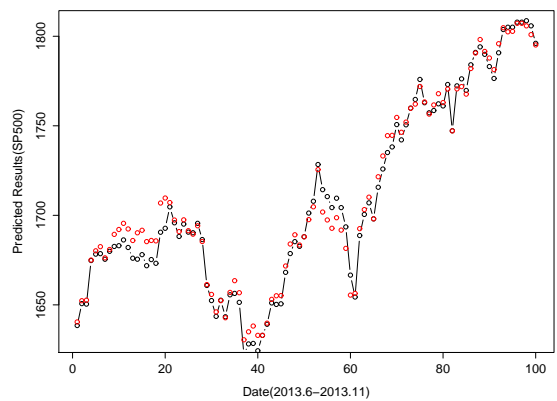}}
\caption{Fitted and predicted results by using SSLS.}
\label{fig:56}
\end{figure}

Furthermore, we implement two methods (SSLS, lasso) and use tracking error to be the measure. We summarize the 18 tracking errors for validation subsets during two years. Each results in Table \ref{table:2} and Table \ref{table:3} omit the percent symbol ($\%$).

See Table \ref{table:2}, SSLS always get lower fitted/predicted errors than lasso. For instance, when SSLS have $2.45\%$ predicted error in Oct. 2013, lasso get $3.95\%$. Furthermore, using SSLS to select 50 stocks, the predicted errors are nearly between $2\%$ and $2.5\%$, which more stable than lasso. By comparison, lasso increase their errors to $2\%$ and $4\%$ when get the same number of nonzero coefficients. Reducing the number of selected stocks to $20$, SSLS also outperforms lasso in almost all the months. The same behavior occurs in the fitted errors.

\begin{table}[!ht]
  \centering
  \setlength\tabcolsep{5pt}
  \caption{The fitted and predicted annual tracking errors obtained by different methods.}
  \label{table:2}
\begin{tabular}{|c||c|cc||c|cc|}
\hline
Methods & Data & Fitted(20) & Fitted(50) & Data & Pred(20)& Pred(50)    \\
    \hline
 SSLS  & 2013.05- &  2.71& 1.07& 2013 & 2.83 & 2.19\\
 Lasso    & 2013.10   &   2.96 & 2.52 & -11& 2.67  & 2.14\\
\hline
 SSLS  & 2013.4- & 3.41 &  1.29& 2013 &  3.30 & 2.50\\
 Lasso    & 2013.9   & 3.47 &  1.92  & -10 & 3.92 & 2.50\\
\hline
 SSLS  & 2013.3- & 2.74 & 0.89 & 2013 & 5.38 & 2.45\\
 Lasso    & 2013.8   & 2.96 &  1.43 & -9 & 5.87 & 3.95\\
\hline
 SSLS  & 2013.2- &  3.52 & 1.16& 2013  & 4.43 & 2.04\\
 Lasso    & 2013.7   &  3.60 & 1.58 & -8 & 3.88 & 2.25\\
\hline
 SSLS  & 2013.1- &  3.03& 1.08 & 2013 & 2.51 & 2.16\\
 Lasso    & 2013.6   &   3.52& 1.73 & -7 & 3.35 & 2.55\\
\hline
 SSLS  & 2012.12- &  2.60& 1.08& 2013 & 4.04 & 2.29\\
 Lasso    & 2013.5   &  4.04 & 1.67 & -6 & 3.96 & 2.50\\
\hline
\end{tabular}
\end{table}

In the Table \ref{table:3}, we compare SSLS and lasso by predicted tracking error in different settings. We consider three situations, selecting 20, 30 and 50 constituent stocks and SSLS always has the better performance. We also find that when we select only 20 stocks in SP500, the predicted error by using SSLS slightly increase but also stable, i.e. $2.71\%$ predicted error in Mar. 2013 and $3.09 \%$ in Aug.2012. At the same time, lasso get $3.79\%$ and $4.16\%$.
\begin{table}[!ht]
  \centering
  \setlength\tabcolsep{10pt}
  \caption{Predicted results under different selected stocks.}
  \label{table:3}
\begin{tabular}{|c||c|c|c|c||c|c|c|c|}
\hline
Methods & Data &  50  & 30& 20 &  Data  & 50  & 30& 20 \\
    \hline
 SSLS  & 2012 & 1.97 & 2.86  & 3.20 & 2012 & 3.60 & 3.68  & 6.54 \\
Lasso    & -6 & 3.09 & 4.04  & 3.90  & -7 & 3.81 & 3.78 & 4.73 \\
\hline
 SSLS  & 2012 & 2.60 & 2.79  & 3.09 & 2012 & 1.88 & 2.53 & 4.11 \\
Lasso    & -8 &   3.52 &  3.56 & 4.16 & -9 &  2.94& 4.23 & 4.95 \\
\hline
 SSLS  & 2012 & 2.46 & 3.15 & 4.44 & 2012 & 2.14 & 3.54 &  4.04 \\
Lasso    & -10 & 2.68 & 3.45  & 6.15 & -11 & 2.34 & 3.51 & 4.25 \\
\hline
 SSLS  & 2012 & 2.75 & 4.04  & 3.92 & 2013 & 2.89 & 3.12 & 4.48 \\
Lasso    & -12 & 3.20 & 4.56  & 5.89 & -1 & 2.87 & 3.62 & 4.55 \\
\hline
 SSLS  & 2013 & 2.50 & 2.32 & 3.30 & 2013 & 2.44 & 2.40 &  2.71\\
Lasso    & -2 & 2.20 & 2.02  & 2.79 & -3 & 2.16 & 2.88 & 3.79 \\
\hline
 SSLS  & 2013 & 2.75 & 3.55  & 5.23 & 2013 & 2.06 & 2.36 & 3.27 \\
Lasso    & -4 & 3.15 &  3.75 & 4.84 & -5 & 1.97 & 2.84 & 5.32 \\
\hline
 SSLS  & 2013 & 2.29 & 2.92 & 4.04 & 2013 & 2.16 & 2.82 & 2.51 \\
Lasso    & -6 & 2.50 &  3.11 & 3.96 & - 7 & 2.25 & 3.23 & 3.35 \\
\hline
 SSLS  & 2013 & 2.04 & 3.55 & 4.43 & 2013 & 2.45 & 4.82 & 5.38\\
Lasso    & -8 & 2.25 &  2.93 & 3.88 & -9 & 3.95 & 4.92 & 5.87\\
\hline
 SSLS  & 2013 & 2.50 & 2.66 &  3.30 & 2013 & 2.19 & 2.70 & 2.83 \\
Lasso    & -10 & 2.50 &  3.21 & 3.92 & -11 & 2.14 & 2.54 & 2.67 \\
\hline
\end{tabular}
\end{table}

As described in Table \ref{table:2} and Table \ref{table:3}, using SSLS and selecting 50 stocks to track SP500, both fitted and predicted annual tracking errors are nearly between $1\%$ and $2\%$. All these results show that SSLS is very successful in assets selection.

\section*{Acknowledgements}
We thank Peter Hall for his helpful comments and suggestions on this paper. This work was supported in part by the National Natural Science Foundation of China (Grant No. 11171361)

\section*{Appendix}
\appendix


First of all, we give follow results to illustrate the property of adaptive elastic net solution without detail proof.
\begin{lem}
  For any $y_n$, $X_n$ in \eqref{eqn:regressionmodel}, the adpative elastic net solution has at most $\min\{n,p\}$ nonzero components as follow
  \begin{equation}
      \hat \beta_{\hat S_n^c}=0,
  \end{equation}
  and
  \begin{equation}
  \label{eqn:solutionafterlemma}
    \hat \beta_{\hat S_n}=(X'_{\hat S_n} X_{\hat S_n}+\lambda_{2,n} I)^{-1}(X'_{\hat S_n} y_n - \lambda_{1,n} \hat w_{\hat S_n} \hat  s_n),
  \end{equation}
  where $\hat S_n$ is defined by
  \begin{equation}
  \hat S_n=\{  i\in \{1,...,p\}: |(X'_{n}(y_n -X_n \hat \beta_n)-\lam_{2,n}\hat \beta_n)_i/ \hat w_{i,n}| =  \lam_{1,n}  \}
  \end{equation}
  and $\hat s_n$ is the corresponding signs. 
\end{lem}

Since the adaptive elastic net penalty function is strictly convex. The solution is always unique, regardless of $X_n$. Similar result for lasso can be seen in \cite{candes2008enhancing,fuchs2004recovery,osborne2000lasso,tib201201,wainwright2009sharp}, the adaptive elastic net solution is given by a simple transformation hence omit proof here.

\begin{lem}
	Consider the linear regression model \eqref{eqn:regressionmodel} with $\epsilon_n$ is a vector of i.i.d. random variables with mean $0$ and variance $\sigma^2$. $X_n$ satisfies (C.1) and (C.2). Given the lasso program \eqref{eq:lasso} with regularization parameter $\lambda_n=4 \sigma (\frac{\log p}{n})^{1/2}$, then there exists constants $c_1,c_2>0$ such that, with probability at least
	$1-o(e^{-n^c})$, any solution $\hat{\beta}_{lasso}$ satisfies the bounds

	\begin{equation}
	||\hat{\beta}_{lasso}-\beta||_1 \leqslant K(\kappa) n^{\eta}.
	\end{equation}
\end{lem}

\begin{proof}[Proof of Lemma 2]

\cite{lv2009unified} gave a similar property for lasso when the noise vector $\epsilon_n$ has i.i.d. $N(0,1)$ entries. Through their results, $l_1$-norm is decomposable when $X_n$ satisfies (C.1) and (C.2) conditions. Also, the choice of $\lam_n$ is given in a similar way. The only difference is to compute the tail bound in the final step. This bound is also used in the proof of Theorem 1. 

By Bernstein's inequality and under condition (C.1) it follows that,
\begin{align}
& P(||X'_n \e/ n||_{\infty} > K n^{\eta}) \leqslant \sum\limits_{j=1}^{p} P(|X'_n \e/ n |> K n^{\eta}) \notag\\
& =  \sum\limits_{j=1}^{p} \exp[ -n^{\frac{2}{3}\eta+\frac{1}{2}} ]   = \exp[ n^c-n^{\frac{2}{3}\eta+\frac{1}{2} }]  =o(e^{-n^{c}}),
\end{align}
completing the proof.
\end{proof}

\begin{pf}
Since 
\begin{equation}
	\hat \beta_n =  \mathop {argmin }\limits_{\beta \in \mR^p}\{\frac{1}{2}||y_n-X_n\beta||^2_2+\frac{1}{2}\lam_{2,n} ||\beta||^2_2 + \lam_{1,n} \sum\limits_{j=1}^{p} \hat w_{j,n} |\beta_j|\}.
\end{equation}

Let $\hat u_n= \sqrt{n} (\hat \beta_n -\beta_n)$ and
\begin{equation}
F_n( \beta_n) =   \frac{1}{2}||y_n-X_n\beta_n||^2_2+\frac{1}{2}\lam_{2,n} ||\beta_n||^2_2 + \lam_{1,n} \sum\limits_{j=1}^{p} \hat w_{j,n} |\beta_{j,n}|.
\end{equation}

Define $V_n(\hat u_n)= F_n (\hat \beta_n)- F_n (\beta_n)$, it can be  written as
\begin{align}
V_n(\hat u_n) & = \frac{1}{2} \hat u_n' C_n \hat u_n - \hat u_n' W_n + \frac{\lam_{2,n}}{2n} \hat u_n' \hat u_n + \lam_{2,n} \frac{\hat u_n'}{\sqrt{n}} \beta_n \notag\\
&+ \lam_{1,n} \sum\limits_{j=1}^{p} \hat w_{j,n} \left( |\beta_{j,n}+ \frac{\hat u_{j,n}}{\sqrt{n}} | -|\beta_{j,n}| \right).
\end{align}

Define $C_n= \frac{1}{n}X'_n X_n$ and $W_n = X'_n \e / \sqrt{n}$. Let $\hat \beta^{(1)}_n$, $\hat \beta^{(2)}_n$ and $W^{(1)}_n$, $W^{(2)}_n$ as the first $q$ and last $p-q$ elements of $\hat \beta_n$ and $W_n$ respectively. Similar, $\hat u_n^{(1)}$ and $\hat u_n^{(2)}$ denote the first $q$ and last $p-q$ elements of $\hat u_n$. Similar as in the proof of Lemma 2, we have
\begin{align}
P(A^c_{1,n}) & = P(||W_n||_{\infty} > K n^{\eta}) \leqslant \sum\limits_{j=1}^{p} P(|W_{j,n} |> K n^{\eta}) \notag\\
& =  \sum\limits_{j=1}^{p} \exp[ -n^{\frac{2}{3}\eta} ]   = \exp[ n^c-n^{\frac{2}{3}\eta} ]  =o(e^{-n^{c}}).
\end{align}

Since $\widetilde{\beta}_n$ is computed by $\hat \beta_{lasso}$ and $\hat w_{j,n} = |\widetilde{\beta}_{j,n}|^{\gamma}$. Follow the result of Lemma 2, we have
\begin{align}
& P(||\hat{\beta}_{lasso}-\beta||_{\infty} \geqslant K(\kappa) n^{\eta})  \\
& \leqslant P(	||\hat{\beta}_{lasso}-\beta||_1 \geqslant K(\kappa) n^{\eta}) = o(e^{-n^{c}}).
\end{align}

Conditioned on $A_n$ and $\{||\hat{\beta}_{lasso}-\beta||_{\infty} \leqslant K(\kappa) n^{\eta}\}$, setting $0< \delta<1$, we have
\begin{align}
\label{eq:V1}
V_n(\hat u_n) & \geqslant ||\hat u_n^{(1)}||\left\{ ||\hat u_n^{(1)}|| \left( \frac{1-\delta}{2}\Lambda_{min}(C_{11,n}) + \frac{\lam_{2,n}}{2n} \right) + \frac{\lam_{2,n}}{\sqrt{n}} ||\beta^{(1)}_n|| -||W^{(1)}_n|| \right\} \notag\\
& - 2 \lam_{1,n} \sum\limits_{j=1}^{q} \frac{|\beta_{j,n}|}{|\widetilde{\beta}_{j,n}|^{\gamma}} + \sum\limits_{j=q+1}^{p} |\hat u_j| \left( \frac{\lam_{1,n}}{\sqrt{n}} \frac{1}{|\widetilde{\beta}_{j,n}|^{\gamma} } - |W_{j,n}| \right) \notag\\
& \geqslant ||\hat u_n^{(1)}||\left\{ ||\hat u_n^{(1)}|| \left( \frac{1-\delta}{2}Kn^{-a} + \frac{\lam_{2,n}}{2n} \right) + \frac{\lam_{2,n}}{\sqrt{n}} ||\beta^{(1)}_n|| -K n^{\eta+b} \right\} \notag\\
& - 2 \lam_{1,n} \sum\limits_{j=1}^{q} |\beta_{j,n}|^{1-\gamma} \left( 1+ K(\gamma,\kappa) n^{(\eta - \frac{1}{2})} \right) + \notag\\
& K(\gamma,\kappa) \sum\limits_{j=q+1}^{p} |\hat u_j| \left( \lam_{1,n} \cdot n^{-\frac{1}{2}} n ^{\frac{\gamma}{2}-\eta \gamma}- n^{\eta} \right).
\end{align}

Following the setting of $\eta$, $\gamma$ and $\lam_{1,n}$, through \eqref{eq:V1}, it follows that $V_n(\hat u_n) > 0$ when $||\hat u_n^{(1)}|| \geqslant M_n$,
\begin{equation}
M_n \equiv K(\delta,\lam_{2,n}) \cdot n^{\eta+a+b}.
\end{equation}

Since $V_n(0)=0$, the minimum of $V_n(\hat u_n)$ can not be attained at $||\hat u_n^{(1)}|| \geqslant M_n$. Then, assume $\{ \hat u_n \in \mR^p: ||\hat u_n^{(1)}|| < M_n, \hat u_n^{(2)} \neq 0 \}$, following inequalities hold uniformly:
\begin{align}
\label{eq:V2}
V_n(\hat u_n) - V_n( \hat u_n^{(1)}, 0) & = \frac{1}{2} (\hat u_n^{(1)})' C_{12,n} \hat u_n^{(2)} + \frac{1}{2} (\hat u_n^{(2)})' C_{22,n} \hat u^{(2)} - (\hat u^{(2)})'W^{(2)}_n  \notag\\
& + \frac{\lam_{2,n}}{2 n} (\hat u_n^{(2)})' \hat u_n^{(2)} +  \lam_{2,n} \frac{(\hat u^{(2)})'}{\sqrt{n}} \beta^{(2)}_n + \frac{\lam_{1,n}}{\sqrt{n}} \sum\limits_{j=q+1}^{p} \frac{|\hat u_{j,n}|}{|\widetilde{\beta}_{j,n}|^{\gamma}} \notag\\
& \geqslant \sum\limits_{j=q+1}^{p}|\hat u_{j,n}| \left[  \dfrac{\lam_{1,n}}{\sqrt{n}} |\widetilde{\beta}_{j,n}|^{-\gamma} -|W_{j,n}| -\frac{1}{2}\left|\left( (\hat u_n^{(1)})'C_{12,n}\right)_j\right|  \right] \notag\\
& \geqslant K \sum\limits_{j=q+1}^{p}|\hat u_{j,n}| \left[   \lam_{1,n} \cdot n^{-\frac{1}{2}} n ^{\frac{\gamma}{2}-\eta \gamma}- n^{\eta} - q^{1/2}\cdot M_n  \right] \notag\\
& >0.
\end{align}

The first inequality of \eqref{eq:V2} holds since $\frac{1}{2} (\hat u_n^{(2)})' C_{22,n} \hat u_n^{(2)} \geqslant 0$, $ \frac{\lam_{2,n}}{2 n} (\hat u^{(2)})' \hat u^{(2)} \geqslant 0$ and $\beta^{(2)}_n=0$. $\left( (\hat u^{(1)})'C_{12,n}\right)_j$ is bounded by $q^{1/2}||u^{(1)}||$. The last inequality holds by the setting of $\eta$,
\begin{equation}
0< \eta < \dfrac{\gamma-1}{2(\gamma+1)} - \dfrac{3b+2a}{2(\gamma+1)}.
\end{equation}

Then the minimum of $V_n(u_n)$ can not be attained at $u_n^{(2)} \neq 0$ too, hence we have the follow result,
\begin{equation}
 \mathop {argmin }\limits_{\hat u_n \in \mR^p} V_n(\hat u_n) \in B_n \equiv \left\{ u_n \in \mR^{p} :  ||\hat u_n^{(1)}|| \leqslant M_n, \hat u_n^{(2)} = 0  \right\},
\end{equation}
completing the proof.
\end{pf}

\begin{proof}[Proof of Corollary 1]
	When $\lam_{1,n} \geqslant K(\delta,\lam_{2,n}) \cdot n^{\eta+a+b}$, we have $V_n(\hat u_n) >0$ if
	\begin{equation}
	||\hat u^{(1)}_n|| \geqslant 3\lam_{1,n} ||\beta^{(1)}_n||,
	\end{equation}
	and hence \eqref{eq:V2} holds if 
	\begin{equation}
	n^{\eta \gamma} < n^{\frac{\gamma-1-b}{2}} ||\beta^{(1)}_n||,
	\end{equation}
	completing the proof
\end{proof}


\begin{pf}
Follow the setting in Lemma 2, $\hat S_n=\{j\in \{1,...,p\}: \hat \beta_{j,n} \neq 0 \}$, we have $\hat \beta_n= \hat \beta_{\hat S_n}$, conditioned on $\{ \hat S_n = S_n \}$, then
\begin{equation}
 \hat \beta_{n}= (X'_{S_n}X_{S_n}+ \lam_{2,n} I)^{-1}(X'_{S_n}y - \lam_{1,n} \hat w_{S_n} s_n),
\end{equation}
where $ s_n=sign(\beta_{S_n})$.

Considering the bias of $\hat \beta_n$, under (C.1) and (C.1) it follows that
\begin{align}
\label{eq:bias}
 ||E\hat \beta_{n} - \beta_n||_2 & \leqslant ||E\hat \beta_{S_n}  1_{\{
\hat S_n =S_n\}}- \beta_n||_2 + ||E\hat \beta_{\hat S_n}  1_{\{
\hat S_n \neq S_n\}}||_2 \notag\\
& \leqslant  ||E\hat \beta_{S_n}  - \beta_{n}||_2 +  ||E\hat \beta_{S_n} 1_{\{
\hat S_n \neq S_n\}}||_2 \notag\\
&+ ||E\hat \beta_{\hat S_n}  1_{\{
\hat S_n \neq S_n\}}||_2.
\end{align}

For every given $\lam_2$, under condition (C.1), the first item of the right hand of \eqref{eq:bias} can be calculated as follow
\begin{align}
||E\hat \beta_{S_n} - \beta_{n}||^2_2  \leqslant & (\Lambda_{min}(C_{11,n})/n +\lam_2)^{-2} \cdot (\lam^2_{2,n} ||\beta_{n}||^2_2  + \lam_{1,n}^2  ||E\hat w_{S_n}||_2^2    ) \notag\\
& \leqslant 2(Kn^{1-a}+\lam_{2,n})^{-2} \cdot (\lam^2_{2,n} ||\beta_{n}||^2_2  + \lam_{1,n}^2  ||E\hat w_{S_n}||_2^2    ),
\end{align}
where $a \in [0,1]$. By Cauchy-Schwarz inequality, the second item can be written as
\begin{align}
 ||E\hat \beta_{S_n}  1_{\{
\hat S_n \neq S_n\}}||_2^2 & \leqslant  E||\hat \beta_{S_n} ||_2^2 P(
\hat S_n \neq S_n) \notag\\
& \leqslant  P(\hat S_n \neq S_n) \cdot  (3||\beta_n||^2_2 +3(X'_{S_n} X_{S_n} + \lam_2 I)^{-2} \cdot \notag\\
& \quad \quad \quad \quad 
(\lam_{2,n}^2 ||\beta_n||^2_2 + \lam_{1,n}^2 ||E \hat w_{S_n}||_2^2+q \cdot n )) \notag\\
& \leqslant   3P(\hat S_n \neq S_n)\cdot (||\beta_n||^2_2 + (Kn^{1-a}+\lam_{2,n})^{-2} \cdot \notag\\
& \quad \quad \quad \quad 
(\lam_{2,n}^2 ||\beta_n||^2_2 + \lam_{1,n}^2 ||E \hat w_{S_n}||_2^2+q \cdot n )).
\end{align}

Setting $|\hat S_n|=d$, follow the result of Lemma 2 in Appendix, we know that $d \leqslant n$., the final item can be written as
\begin{align}
 ||E\hat \beta_{\hat S_n}  1_{\{
\hat S_n \neq S_n\}}||_2^2 & \leqslant  E||\hat \beta_{\hat S_n} ||_2^2 P(
\hat S_n \neq S_n) \notag\\
& \leqslant 3P(\hat S_n \neq S_n) \cdot (||\beta_n||^2_2 +(X'_{\hat S_n} X_{\hat S_n, n}+ \lam_{2,n} I)^{-2} \cdot \notag\\
& \quad \quad \quad \quad 
(\lam_{2,n}^2 ||\beta_n||^2_2 + \lam_{1,n}^2 E||\hat w_{n}||^2_2+d \cdot n)) \notag\\
& \leqslant 3P(\hat S_n \neq S_n)\cdot (||\beta_n||^2_2 + (Kn^{1-a}+\lam_{2,n})^{-2} \cdot \notag\\
& \quad \quad \quad \quad 
(\lam_{2,n}^2 ||\beta_n||^2_2 + \lam_{1,n}^2 E ||\hat w_{n}||_2^2 +n^2)).
\end{align}

Combining the above results, we obtain the bias of $\hat \beta_n$ as follow
\begin{align}
||E\hat \beta_{n} - \beta_n||_2^2 \leqslant & 2[1+ 3 P(\hat S_n \neq S_n)] \cdot (Kn^{1-a}+\lam_{2,n})^{-2} \cdot \notag\\
&(\lam_{2,n}^2 ||\beta_n||^2_2 + \lam_{1,n}^2 E||\hat w_{n}||^2_2) \cdot \notag\\
&+ 6 P(\hat S_n \neq S_n) \cdot (||\beta_n||^2_2+ n( n \vee p)).  
\end{align}

Next, we proof the MSE of the adaptive elastic net estimator
\begin{align}
\label{eq:MSE}
& E||\hat \beta_{n}- \beta_n||^2_2  \notag\\
& = E||\hat \beta_{n}- \beta_n||^2_2 1_{\{\hat S_n=S_n\}}+ E||\hat \beta_{n}- \beta_n||^2_2 1_{\{\hat S_n \neq S_n\}} \notag\\
& \leqslant E||\hat \beta_{S_n}- \beta_n||^2_2 1_{\{\hat S_n=S_n\}} + 2( E ||\hat \beta_{n}||_2^2  1_{\{\hat S_n \neq S_n\}} + E||\beta_n||  1_{\{\hat S_n \neq S_n\}} ) \notag\\
& \leqslant   3(Kn^{1-a}+\lam_{2,n})^{-2} \cdot 
(\lam_{2,n}^2 ||\beta_n||^2_2 + \lam_{1,n}^2 E|| \hat w_{S_n}||_2^2+q \cdot n ) \notag\\
&\quad + 2 \sqrt{P(\hat S_n \neq S_n)} (||\beta_{n}||^2_2 + \sqrt{E||\hat \beta_{n}||^4_2}).
\end{align}

$E||\hat \beta_n||^4_2$ satisfies
\begin{align}
\sqrt{E||\hat \beta_n||^4_2} & \leqslant 3 (||\beta_n||^2_2 + (Kn^{1-a}+\lam_{2,n})^{-2} \cdot \notag\\
& \quad \quad \quad \quad 
(\lam_{2,n}^2 ||\beta_n||^2_2 + \lam_{1,n}^2 E ||\hat w_{n}||_2^4 +n^2 )).
\end{align}

Therefore, if $n$ is large enough, \eqref{eq:MSE} can be written as
\begin{align}
E||\hat \beta_{n}- \beta_n||^2_2  \leqslant & 3[1+ 2 \sqrt{P(\hat S_n \neq S_n)}] \cdot (Kn^{1-a}+\lam_{2,n})^{-2} \cdot \notag\\
&(\lam_{2,n}^2 ||\beta_n||^2_2 + \lam_{1,n}^2 E||\hat w_{n}||^4_2+q \cdot n) \cdot \notag\\
&+ 8 \sqrt{P(\hat S_n \neq S_n)} \cdot (||\beta_n||^2_2+n^2),
\end{align}
which completes the proof.
\end{pf}


\begin{pf}
Setting
\begin{equation}
\label{equ:U1}
U^{(1)}_n = C^{-1}_{11,n}(\lam_{2,n})(W^{(1)}_n - n^{-1/2} \widetilde{s}^{(1)}_n  ),
\end{equation}
where $\widetilde{s}^{(1)}_n= (\widetilde{s}_{1,n},...,\widetilde{s}_{q,n})'$ with $\widetilde{s}_{j,n} \equiv sign (\beta_{j,n}) |\widetilde{\beta}_{j,n}|^{-\gamma}$, $1 \leqslant j \leqslant q$, and $C_{n}(\lam_{2,n})= \frac{1}{n} (X'_{n}X_{n}+ \lam_{2,n}I )$. \eqref{equ:U1} is the first $q$ elements of adaptive elastic net estimator.

Obtained by Theorem 1, we have
\begin{equation}
 \sqrt{n} (\hat \beta_n -\beta_n) =\mathop {argmin}\limits_{||\hat u_n^{(1)},0||}  V_n( \hat u_n^{(1)}, 0) = (U'^{(1)}_n,0')'
\end{equation}
Setting $T_n =D^{(1)}_nU_n^{(1)}$. By Taylors expansion and EE expansion, setting 

$Q_n=D^{(1)}_n (C^{-1}_{11,n}(\lam_{2,n})C_{11,n}C^{-1}_{11,n}(\lam_{2,n})) (D_n^{(1)})'$, $\Psi_{1,n} (B) = \int_B \psi_{1,n} (x)	dx$ and $\psi_{1,n}(x)$ is the Lebesgue density of the Edgeworth expansion for $T_{1,n}$.

We have
\begin{align}
\label{eqn:thm3}
& \sup\limits_{B \in C} | P(T_n \in B)- \Phi (B, \sigma^2 Q_n )| \notag\\
& \leqslant \sup\limits_{B \in C} | P(T_{n} \in B)- P(T_{1,n} \in B)|  + \sup\limits_{B \in C} | P(T_{1,n} \in B)- \Phi (B, \sigma^2 T_n )| \notag\\
& \leqslant \sup\limits_{B \in C} | P(T_{n} \in B)- P(T_{1,n} \in B)|  + \sup\limits_{B \in C} | P(T_{1,n} \in B)- \Psi_{1,n} (B)  | \notag\\
& +  \sup\limits_{B \in C}  |\Psi_{1,n} (B) - \Phi (B: \sigma^2 Q_n )| \notag\\
& \leqslant O(n^{-1/2} +||b_n|| + \lam_n n^{a+e(\gamma+1)-1}),
\end{align}
where $T_{1,n}$ is the Taylor's expansion of $T_n$ and $T_n -T_{1,n}$ is the remainder term obtained by Taylor's expansion. Therefore $||T_n -T_{1,n}||$ have bounds $o(n^{-1/2})$ and hence the first item of \eqref{eqn:thm3} is bounded by $ o(n^{1/2})$. The second inequality of \eqref{eqn:thm3} holds after calculations and bounds by (C.1), (C.3) and (C.4). More details can be seen in Bhattacharya and Ranga Rao(1986).

Setting $\sigma^2 \widetilde{Q}_n$ be the variance of $T_{1,n}$, $R^{(1)}_n$ be a $q \times q$ diagonal matrix with $j$th diagonal entry given by $sgn(\beta_{j,n})|\beta_{j,n}|^{-(\gamma +1)}$, $1 \leqslant j \leqslant q$. Under conditions(C.1), (C.3) and (C.4), we have

\begin{align}
& ||\widetilde{Q}_n-Q_n|| \leqslant \frac{\gamma \lam_{1,n}}{n} ||D_n^{-1} C^{-1}_{11,n} (\lam_2) R^{(1)}_n C^{-1}_{11,n} (\lam_{2,n})( D_n^{-1})'|| \notag\\
& + K \cdot n^{-1} \left( \frac{\lam^2_{1,n}}{n^2} \sum\limits_{i=1}^{n} \sum\limits_{j=1}^{q} |\beta_{j,n}|^{-2(\gamma+1)} \right)||D_n^{(1)} C^{-1}_{11,n}(\lam_{2,n})||^2 \notag\\
& \leqslant K \cdot\frac{\lam_{1,n}}{n} ||C^{1/2}_{11,n}(\lam_{2,n})||^2 \left( ||R^{(1)}_n|| + \frac{\lam_{1,n}}{n} \sum\limits_{j=1}^{q} |\beta_{j,n}|^{-2(\gamma+1)} \right) \notag\\
&  = O\left(\frac{\lam_{1,n}}{n} n^{a+e(\gamma+1)} \right).
\end{align}

Hence the final item of \eqref{eqn:thm3} holds
\begin{align}
& \sup\limits_{B \in C}  |\Psi_{1,n} (B) - \Phi (B: \sigma^2 Q_n )| \notag \\
& \leqslant \sup\limits_{B \in C}  |\int_{B} \phi(x, \sigma^2 \widetilde{Q}_n) - \phi (x, \sigma^2 Q_n )| + O(||b_n||)\notag \\
& \leqslant O(\lam_n n^{a+e(\gamma+1)-1}+||b_n||),
\end{align}
where $\phi(x, \sigma^2 Q_n)$ denotes the density of $N(0, \sigma^2 Q_n)$, and
\begin{equation}
||b_n|| \leqslant \frac{\lam_{1,n}}{\sqrt{n}} ||D_n^{(1)}C^{1/2}_{11,n}(\lam_{2,n})|| \cdot ||C^{1/2}_{11,n}(\lam_{2,n})|| \cdot ||s^{(1)}_n|| =O(n^{-\delta}),
\end{equation}
where is the part of the second item of Edgeworth expansion for $T_n$, completing the proof.

\end{pf}

\begin{pf}
Conditioned on $\{ \hat S_n=S_n  \}$, the SSLS estimator $\hat \beta_n(ssls)$ satisfies
\begin{align}
\hat \beta_n(ssls) & =  \left( X'_{S_n} X_{S_n} \right)^{-1}X'_{S_n}y \notag\\
& = \beta_{n} + \left( X'_{S_n} X_{S_n} \right)^{-1}X'_{S_n} \e .
\end{align}

Therefore
\begin{align}
& P( \alpha' \Sigma^{1/2}_{S_n}( \hat \beta_{n} (ssls) - \beta_{n} ) \leqslant t ) \notag\\
& \leqslant P ( \alpha' \Sigma^{1/2}_{S_n} ( \hat \beta_{S_n} (ssls) - \beta_{n} ), \hat S_n = S_n ) + P ( \hat S_n \neq S_n ) \notag\\
& \leqslant  P ( \alpha' \Sigma^{-1/2}_{S_n}X'_{S_n}\e )  + 2P( \hat S_n \neq S_n ),
\end{align}
and
\begin{equation}
2 P(\hat S_n \neq S_n) = o(e^{-n^{c}}) \rightarrow 0,  \ as \ n \rightarrow \infty.
\end{equation}

Write $r_i= \alpha' \Sigma^{-1/2}_{S_n}X'_{\cdot, i}$ where $i=1,...,n$ and $X'_{\cdot,i} \in \mR^q$, by Lyapunov conditions for the central limit theorem, we have
\begin{align}
E(\alpha' \Sigma^{-1/2}_{S_n}X'_{S_n}\e)^{2+\delta} & = \sum\limits_{i=1}^n E|\e_i|^{2+\delta}\cdot|r_i|^{2+\delta} \notag\\
& \leqslant E|\e|^{2+\delta}(\sum\limits_{i=1}^n|r_i|^2 (\max\limits_{i}|r_i|^{\delta})) \notag\\
& = E|\e|^{2+\delta} (\max\limits_{i}|r_i|^{2})^{\delta/2},
\end{align}
and
\begin{equation}
r_i^2 \leqslant 2\Sigma_{S_n}^{-1} \cdot \sum\limits_{j=1}^{q}x_{ij}^2,
\end{equation}
follow the condition in Theorem \ref{thm:4}, completing the proof.
\end{pf}

\begin{pf}

Conditional on $\{\hat S_n=S_n \}$, the SSLS estimator can be written as
\begin{equation}
\hat \beta_n(ssls)= ( X'_{S_n} X_{S_n} )^{-1} X_{S_n} y.
\end{equation}

Therefore
\begin{align}
\label{eq:ssls}
||E\hat \beta_n (ssls)- \beta_n||_2
 & \leqslant ||E\hat \beta_{S_n}(ssls)- \beta_{n}||_2
  + ||E \hat \beta_{S_n} (ssls) 1_{\{ \hat S_n \neq S_n \}}||_2 \notag\\
  & + ||E \hat \beta_{n} (ssls) 1_{ \{ \hat S_n \neq S_n \}}||_2.
\end{align}

Considering about the right hand of \eqref{eq:ssls},
\begin{equation}
E \hat \beta_{S_n} (ssls)=\beta_{S_n},
\end{equation}
and under condition (C.1) it follows that
\begin{align}
||E \hat \beta_{S_n} (ssls) 1_{ \{\hat S_n \neq S_n \}}||_2^2 & \leqslant E|| \hat \beta_{S_n} (ssls) ||_2^2 P( \hat S_n \neq S_n ) \notag\\
& =P( \hat S_n \neq S_n ) (  ||\beta_{n}||_2^2 + \frac{\sigma^2}{n} \cdot q \cdot \Lambda_{min}^{-1} (C_{11,n})   ) \notag\\
& \leqslant P( \hat S_n \neq S_n ) (  ||\beta_{n}||_2^2 + \sigma^2K^{-1} n^{a+b-1} ).
\end{align}

Setting $|\hat S_n|=d$, follow the result of Lemma 2, $d \leqslant n$ then
\begin{align}
||E \hat \beta_{n} (ssls) 1_{ \left\{\hat S_n \neq S_n \right\}}||^2_2 & \leqslant  E||\hat \beta_n(ssls)||^2_2  P(  \hat S_n \neq S_n )  \notag\\
& \leqslant P(\hat S_n \neq S_n) ( ||\beta_n||^2_2 + \frac{\sigma^2}{n} \cdot d \cdot K^{-1}n^a ) \notag\\
& \leqslant P(\hat S_n \neq S_n) ( ||\beta_n||^2_2 + \sigma^2  \cdot K^{-1}n^a ).
\end{align}

So the bias of SSLS estimator is bounded by
\begin{equation}
||E\hat \beta_n (ssls)- \beta_n||^2_2 \leqslant 2 P ( \hat S_n \neq S_n  ) ( 2||\beta_n||^2_2 + \sigma^2K^{-1} n^{a-1} + \sigma^2 K^{-1} n^{a-1} n \vee q ) .
\end{equation}

Considering the MSE of SSLS estimator, we have
\begin{align}
\label{eq:MSE2}
E||\hat \beta_n (ssls)- \beta_n||^2_2 & \leqslant E||\hat \beta_n (ssls)- \beta_n||^2_2  1_{ \{\hat S_n = S_n \}} +   E||\hat \beta_n (ssls)- \beta_n||^2_2 1_{ \{\hat S_n \neq S_n \}} \notag\\
& \leqslant \sigma^2K^{-1} n^{a+b-1} + 8 \sqrt{P(\hat S_n \neq S_n)} (||\beta_n||^2_2 + \sigma^2  \cdot K^{-1}n^a ).
\end{align}

The last inequality of \eqref{eq:MSE2} holds by the similar calculations of Theorem 2, which completes the proof.
\end{pf}

\end{document}